\documentclass [letterpaper] {patmorin}
\usepackage{amsfonts}
\usepackage{amssymb}
\usepackage{amsmath}
\usepackage{amsthm}
\usepackage{graphicx}
\usepackage[mathlines]{lineno}
\newtheorem{theorem}{Theorem}

\newtheorem{lemma}{Lemma}

\newtheorem{op}{Open Problem}

\newtheorem{corollary}{Corollary}

\newcommand{\EXP}{\mathrm{E}}

\newcommand{\R}{\mathbb{R}}

\newcommand{\etal}{\emph{et al}}
\newcommand{\ccw}{\mathrm{ccw}}
\newcommand{\cw}{\mathrm{cw}}

\title{\MakeUppercase{Memoryless Routing in Convex Subdivisions: \newline 
       Random Walks are Optimal}}
\author{Dan Chen, Luc Devroye, Vida Dujmovi\'c, and Pat Morin}

\begin{document}
\maketitle

\begin{abstract}
A \emph{memoryless} routing algorithm is one in which the decision
about the next edge on the route to a vertex $t$ for a packet currently
located at vertex $v$ is made based only on the coordinates of $v$, $t$,
and the neighbourhood, $N(v)$, of $v$.  The current paper explores the
limitations of such algorithms by showing that, for any (randomized)
memoryless routing algorithm $\mathcal{A}$, there exists a convex
subdivision on which $\mathcal{A}$ takes $\Omega(n^2)$ expected
time to route a message between some pair of vertices.  Since this
lower bound is matched by a random walk, this result implies that the
geometric information available in convex subdivisions is not helpful
for this class of routing algorithms.  The current paper also shows
the existence of triangulations for which the \textsc{Random-Compass}
algorithm proposed by Bose \etal\ (2002,2004) requires $2^{\Omega(n)}$
time to route between some pair of vertices.
\end{abstract}

\section{Introduction}
\label{sec:intro}

In recent years, motivated primarily by the proliferation of wireless networks and GPS devices, much research has been done on routing algorithms for geometric networks \cite{gior03}.  In this research a network is modelled as a geometric graph $G=(V,E)$ whose vertex set $V$ is a set of points in $\R^2$. We say that a routing algorithm $\mathcal{A}$ \emph{works} for $G$ if, for any pair of vertices $s,t\in V$, the algorithm always find a path from $s$ to $t$ in a finite number of steps.

The research on geometric routing algorithms largely focuses on utilizing geometric properties of a class of geometric graphs to reduce the complexity of, and information required by, routing algorithms.  For example, when $G$ is the unit disk graph\footnote{The \emph{unit disk graph} of a point set $V$ contains the edge $uv$ if and only if the Euclidean distance between $u$ and $v$ is at most 1.} of the points in $V$, then an algorithm, called \textsc{Face-1}, of Bose \etal\ \cite{bose01} (see also Karp and Kung \cite{kk00}) works for $G$ and requires no preprocessing of $G$ or additional state information at the vertices of $G$ and requires only a constant size header associated with each packet.  An extremely general result in this vein, based on logspace construction of universal exploration sequences, shows that, using a header containing only $O(\log n)$ bits, one can visit all the vertices of any graph (and hence reach $t$) in a polynomial number of steps \cite{b08}.

A particularly interesting and restricted class of routing algorithms are so-called memoryless routing algorithms.  A \emph{memoryless} routing algorithm is one in which the decision about the next edge on the route to $t$ for a packet currently located at node $v$ is based only on the coordinates of $v$, $t$, and the neighbourhood, $N(v)$, of $v$. More precisely, a deterministic memoryless routing algorithm is a function $f:\R^2\times\R^2\times(\R^2)^+\rightarrow \R^2$ that satisfies $f(v,t,N(v)) \in N(v)$ and $f(t,t,N(t)) = t$ for all inputs.

Note that a memoryless routing algorithm makes each routing step without using information obtained in previous routing steps and without any global information about $G$.  Memoryless algorithms are different from \emph{oblivious} routing algorithm \cite[Section~4.2]{mr95} which select a path from $s$ to $t$ having total knowledge of $G$ (but without knowledge of other source/destination pairs).
 
Bose and Morin \cite{bose04} show that if $G$ is Delaunay triangulation\footnote{A \emph{triangulation} is a geometric graph all of whose faces, except the outer face, are triangles, and whose outer face is the complement of a triangle. Delaunay triangulations and regular triangulations are special types of triangulations. For details, consult \cite{obs92}.} or a regular triangulation then deterministic memoryless routing algorithms, named \textsc{Greedy} and \textsc{Compass}, respectively, work for $G$. Bose \etal\ \cite{bose02} subsequently show a stronger result; a deterministic memoryless routing algorithm, named \textsc{Greedy-Compass}, works for any triangulation $G$.

Memoryless routing algorithms are so simple, elegant, and practical that researchers have spent considerable effort designing geometric embeddings of graphs so that memoryless routing algorithms can be applied to the resulting embeddings.  A famous example in this vein is due to Leighton and Moitra \cite{lm08} who prove that every 3-connected planar graph $\tilde G$ admits an embedding $G$ in $\R^2$ such that \textsc{Greedy} works on $G$.  The combination of the embedding and routing algorithm represents a form of \emph{compact routing} \cite{l94}.

Unfortunately, deterministic memoryless routing algorithms have severe limitations.  These stem from the fact that these algorithms can not visit the same vertex more than once without looping forever.  Bose \etal\ \cite[Theorem~2]{bose04} show that there exists 17 convex subdivisions\footnote{A \emph{convex subdivision} is a geometric graph all of whose faces, except the outer face, are convex polygons, and whose outer face is the complement of a convex polygon.}, $G_1,\ldots,G_{17}$, each with 17 vertices such that any deterministic memoryless routing algorithm does not work for at least one of these subdivisions.  Thus, convex subdivisions form a class of geometric graphs that are too rich for deterministic memoryless routing algorithms \cite{bose02}.

The same authors \cite{bose02,bose04} observe that randomization can be used to overcome this limitation.  A \emph{randomized memoryless} routing algorithm is one in which the decision about the next edge on the route to $t$ for a packet currently located at node $v$ is based only on $v$, $t$, the neighbourhood, $N(v)$, of $v$, and a sequence $B$ of fresh random bits.  More precisely, a randomized memoryless routing algorithm is defined by a function $f:\R^2\times\R^2\times(\R^2)^+\times\{0,1\}^\infty\rightarrow \R^2$ that satisfies $f(v,t,N(v),B) \in N(v)$ and $f(t,t,N(v),B) = t$ for all inputs.  The final argument $B$ is a sequence of random bits that are chosen fresh for each step taken by the routing algorithm. Bose \etal\ describe a randomized memoryless algorithm, named \textsc{Random-Compass}, that uses one random bit per step works for any convex subdivision. They do not analyze the efficiency of \textsc{Random-Compass} except to note that, for some convex subdivisions $G$, and some pairs $s,t\in V$, the expected number of steps taken by \textsc{Random-Compass} when routing from $s$ to $t$ is $\Omega(|V|^2)$.  

Observe that, by the theory of random walks (c.f. \cite[Theorem~6.6]{mr95}), the expected time required for a random walk on $G$ to travel from a particular vertex $s$ to a particular vertex $t$ is $O(n^2)$.  Therefore, a random walk is at least as efficient, in the worst case, as the \textsc{Random-Compass} algorithm. Nevertheless, one might expect that \textsc{Random-Compass} is more likely to find short routes, since it uses geometry to find a route that is specifically directed towards the target vertex $t$.  Thus, we might intuit that \textsc{Random-Compass} is a heuristic that is usually better than a random walk and never much worse.

In the current paper, we show that this intuition about \textsc{Random-Compass} could not be further from the truth.  Indeed, for any $n>0$, there exists a convex subdivision (in fact, a triangulation) $G$ with $n$ vertices and having two vertices $s$ and $t$ such that the expected number of steps taken by \textsc{Random-Compass} when routing from $s$ to $t$ is $2^{\Omega(n)}$. This triangulation has diameter 3.

Next we study whether \emph{any} randomized memoryless routing algorithm for convex subdivisions can outperform a random walk.  We show that, for any randomized memoryless routing algorithm $\mathcal{A}$ and any $n$, there exists a convex subdivision $G=G(\mathcal{A})=(V,E)$ of size $n$ and a pair of vertices $s,t\in V$ such that the expected number of steps taken by $\mathcal{A}$ when routing from $s$ to $t$ is $\Omega(n^2)$.  Therefore, at least in the worst-case, no algorithm significantly outperforms a random walk.

\section{A Bad Example for Random-Compass}

The \textsc{Random-Compass} algorithm works by using a coin toss to select among the (at most two) neighbours $\ccw_t(v)$ and $\cw_t(v)$ of the current node $v$ that make the minimum and maximum angle, respectively, with the segment $vt$ (see Figure~\ref{fig:random-compass}.a). When applied on a convex subdivision $G=(V,E)$, Bose \etal\ show that, in the directed graph $G'$ that contains the edges $(v,\cw_t(v))$ and $(v,\ccw_t(v))$ for all $v\in V$, there exists at least one directed path $P(v,t)$ from every vertex $v$ to $t$ (see Figure~\ref{fig:random-compass}.b).  This, and Wald's Equation, immediately imply that the expected time to reach $t$ from any vertex is at most $2^{n}$; from any vertex $v$, \textsc{Random-Compass} has probability at least $1/2^{|P(v,t)|}\ge 1/2^{n-1}$ of reaching $t$ by following $P(v,t)$, and the expected number of steps it takes on $P(v,t)$ before falling off $P(v,t)$ is at most 2.

\begin{figure}
  \begin{center}
    \begin{tabular}{cc}
    \includegraphics{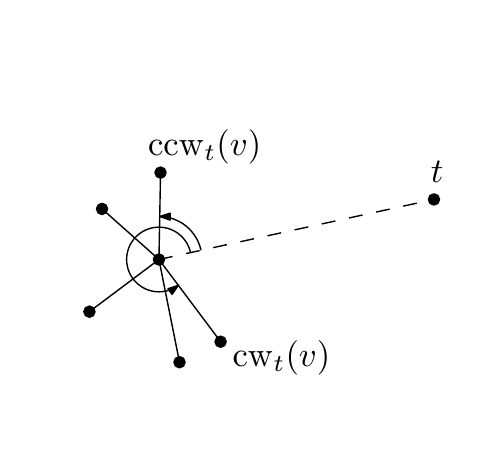} &
    \includegraphics{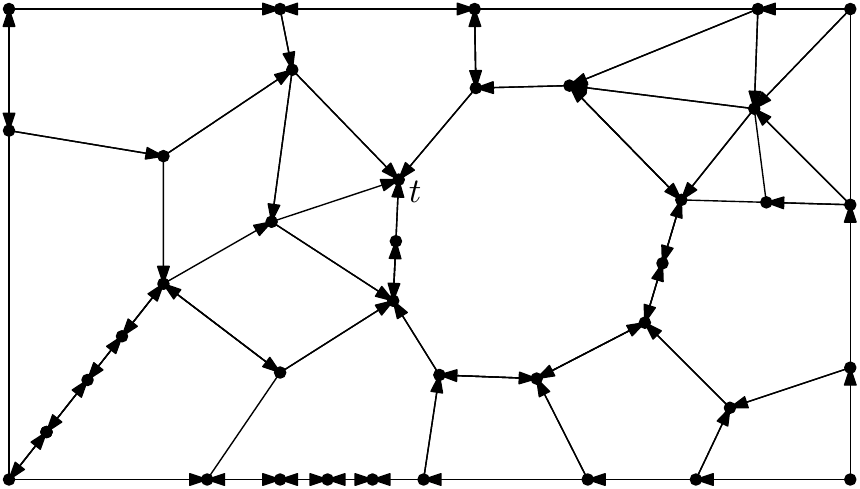} \\
    (a) & (b)
    \end{tabular}
  \end{center}
  \caption{The \textsc{Random-Compass} algorithm chooses the next vertex at random among $\ccw_t(v)$ and $\cw_t(v)$.}
  \label{fig:random-compass}
\end{figure}

The example in Figure~\ref{fig:bad-unbiased} shows that the above analysis of \textsc{Random-Compass}, although very coarse, is about the best one can do. It shows a geometric graph $G$ whose vertex set has size $n=4k+1$ and whose vertices are organized as a central vertex $t$ and four paths leading from the outer face to $t$.  The space between these paths is triangulated so that, at any point, \textsc{Random-Compass} chooses between an edge that leads one step closer to $t$ or that returns to the outer face.

\begin{figure}
  \begin{center}
    \includegraphics{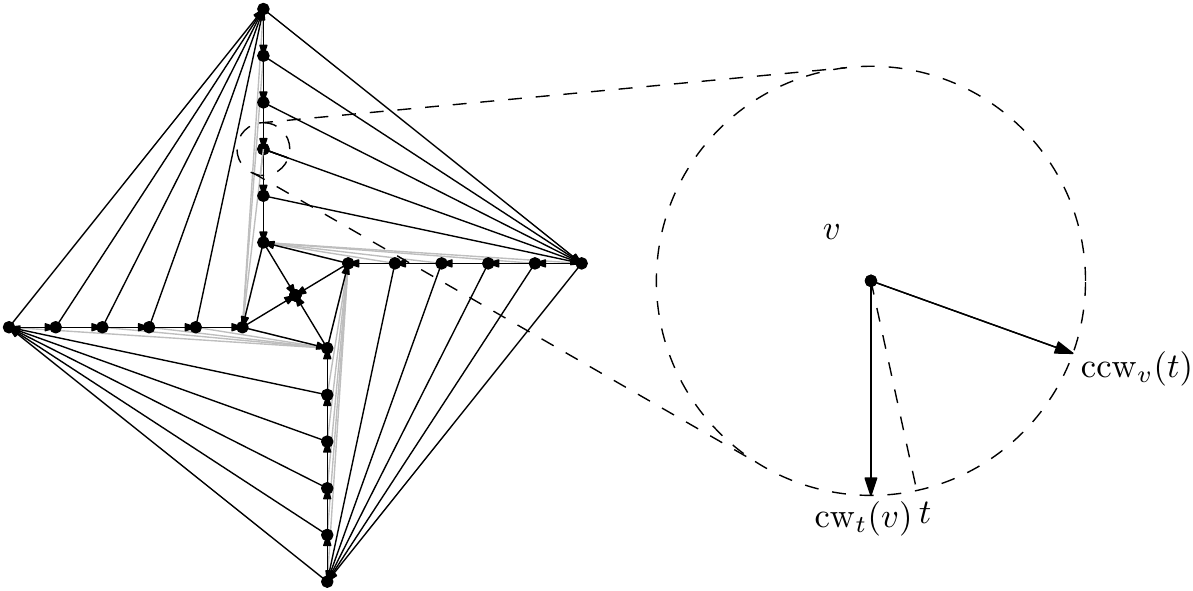}
  \end{center}
  \caption{A graph in which \textsc{Random-Compass} has expected running time 
           $\Omega(2^{n/4})$.}
  \label{fig:bad-unbiased}
\end{figure}

If we consider the directed graph $G'$ defined above, then we see that, at any point the packet is at some distance $i$ from $t$ and that, it can, with equal probability, move to a vertex  of distance $i-1$ or move to a vertex (on the outer face) of distance $k$.  If we denote by $T_i$ the expected number of steps required by \textsc{Random-Compass} to reach $t$ given that it is currently at distance $i$ from $t$, we see that 
\[ 
  T_i = \left\{ \begin{array}{ll}
          0 & \mbox{for $i=0$} \\
          (1/2)T_{i-1} + (1/2)T_k & \mbox{for $i\in\{1,\ldots,k\}$} 
        \end{array}\right.
\]
Expanding the value of $T_k$ gives
\begin{eqnarray*}
  T_k & = & 1 + (1/2)T_k + (1/2)T_{k-1} \\
      & = & 1 + (1/2)T_k + 1/2 + (1/4)T_{k} + (1/4)T_{k-2} \\
      & = & 1 + (1/2)T_k + 1/2 + (1/4)T_{k} + (1/4)
              + \cdots  + (1/2^{k-1}) + (1/2^{k})T_{k} + (1/2^{k})T_0 \\
      & = & 2-1/2^{k} + (1-1/2^k)T_k \enspace ,
\end{eqnarray*}
and rewriting this gives $T_k = 2^k(2-1/2^{k}) = \Omega(2^{n/4})$.  This proves:

\begin{theorem}\label{theorem:random-compass}
For any $n>1$, there exists a triangulation $G$ having two vertices $s$ and $t$ such that the expected number of steps taken by $\textsc{Random-Compass}$ when routing from $s$ to $t$ is $2^{\Omega(n)}$.
\end{theorem}

Note that the base in the exponent can be improved by using a construction with 3 paths instead of 4.  In this case, the lower bound becomes $\Omega(2^{n/3})$.  Furthermore, up to a factor of 2, the lower bound on Theorem~\ref{theorem:random-compass} holds for all choices of the source vertex $s$ since, for any vertex $s\neq t$, the expected time to route from $s$ to $t$ is at least $(1/2)T_k$.

\section{A Lower Bound for Any Algorithm}

In this section we develop an $\Omega(n^2)$ lower bound for routing on convex subdivisions using any randomized memoryless routing algorithm $\mathcal{A}$.  The outline of the lower bound is as follows:  We start with a lemma about Markov chains whose transition graphs are paths. We show that, when starting at the midpoint of the path, there is at most one endpoint of the path that can be reached in subquadratic expected time.  This lemma is relevant since, if $\mathcal{A}$ finds itself in the interior of a path of degree 2 vertices in $G$, it will behave like such a Markov chain until it reaches one of the endpoints of this path.

Next, we observe how $\mathcal{A}$ behaves on certain paths of degree 2 vertices and show that, because $\mathcal{A}$ can only reach one endpoint of any path in subquadratic time, that we can always find a subset of these paths that can be pieced together to form a convex subdivision in which $\mathcal{A}$ takes at least quadratic expected time to route from some vertex $s$ to some vertex $t$.

\subsection{Markov Chains}
\label{sec:markov}

Consider a Markov chain on $\{ 1, \ldots, n \}$, $n > 1$, where transitions only
take place between neighbors. If $p_{i,j}$ is the probability of a transition
from $i$ to $j$, then we have 
\[
p_{1,2} = p_{n,n-1} = 1,
\]
\[
p_{i,i+1} = 1-p_{i,i-1} = \pi_i, 2 \le i \le n-1,
\]
where $\pi_2, \ldots, \pi_{n-1}$ are fixed probabilities.
The vector of these probabilities is denoted by $\pi$.
We will set $\pi_1 = 1$, $\pi_n = 0$, to be consistent, as the
extreme states are reflecting.  When $\pi_i = 1/2$ for $2 \le i \le n-1$,
we obtain a standard random walk on a finite interval with reflecting barriers.

We denote the Markov chain by $X_0, X_1, \ldots, X_t , \ldots$,
and denote the hitting times by $T_{i,j}$:
\[
T_{i,j} = \min \{ t > 0: X_t = j | X_0 = i \}.
\]
For a standard random walk, it is known that
\[
\EXP \{ T_{i,j} \} = (j-i)^2, j \not= i, 1 \le j,i \le n
\]
\cite{m73}.  The standard random walk is in fact the best possible chain in the following sense:

\begin{lemma}\label{lemma:backandforth}
For any vector of probabilities $\pi$, and any $n > 1$,
\[
\EXP \{ T_{1,n} + T_{n,1} \} \ge 2 (n-1)^2.
\]
\end{lemma}

\begin{proof}
The lemma is obviously true if any $\pi_i$, $2 \le i \le n-1$, is either
zero or one as that would imply that at least one of the hitting times is
infinite. Thus, we assume that all probabilities are strictly in $(0,1)$.
It is also trivial if $n=2$, so assume $n > 2$.
Define
\[
P_i = {1 \over \pi_i } - 1, Q_i = {1 \over 1-\pi_i} - 1,
\]
and note that $P_i Q_i = 1$.
If needed, we formally set $P_1 = Q_n = 0$.
%

We need an explicit formula for $\EXP \{ T_{1,1} \}$.
Let us introduce the chains on $\{ i, \ldots , n \}$ with
reflecting barriers at $i$ and $n$, but with the same $\pi_j$ values
associated with non-terminal states. 
Let $T^+_{i,j}$ with $j \ge i$, denote the hitting time from $i$ to $j$ in the chain $\{1,\ldots,n\}$ defined this way.
Clearly,
\[
T^+_{n-1, n-1} = 2.
\]
Next,
\[
T^+_{n-2, n-2 } = 2 + \sum_{j \le Z} W_j,
\]
where $W_j$ are independent lengths excursions from $n-1$ to  $n-1$
on the chain $\{ n-1, n \}$, and $Z$ (possibly zero) is the number of such excursions.
Obviously, $Z$ is geometrically distributed, and $\EXP \{ Z \} = Q_{n-1}$.
Because $\EXP \{ W_1 \} = \EXP \{ T^+_{n-1, n-1} \} = 2$,
and because $Z$ is a stopping time, 
we have, by Wald's identity,
\[
\EXP \{ T^+_{n-2, n-2} \} = 2 + 2 Q_{n-1}.
\]
This argument is easily extended by induction, and we obtain for $1 \le i < n-1$,
\[
\EXP \{ T^+_{i, i} \}  = 2 + Q_{i+1} \EXP \{ T^+_{i+1, i+1} \}
  = 2 \left( 1 + Q_{i+1} + Q_{i+1} Q_{i+2} + \cdots + Q_{i+1} \cdots Q_{n-1} \right).
\]
By flipping sides, and denoting by $T^-$ the hitting times for
the Markov chains on $\{ 1, \ldots i \}$
with reflecting bariers at $1$ and $i$, we obtain in a similar fashion,
for $2 < i \le n$,
\[
\EXP \{ T^-_{i, i} \}  = 2 + P_{i-1} \EXP \{ T^-_{i-1, i-1} \}
  = 2 \left( 1 + P_{i-1} + P_{i-1} P_{i-2} + \cdots + P_{i-1} \cdots P_{2} \right).
\]
Furthermore, $\EXP \{ T^-_{2, 2} \}  =2$.

With these calculations out of the way,  we note that
\[
\EXP \{ T_{i,i+1} \} = 1 + P_i \times \EXP \{ T^-_{i, i} \} \enspace ,
\]
and
\[
\EXP \{ T_{i,i-1} \} = 1 + Q_i \times \EXP \{ T^+_{i, i} \} \enspace .
\]
Clearly,
\[
\begin{aligned}
\EXP &\{  T_{1,n} + T_{n,1} \} \\
&= \sum_{i=1}^{n-1} \EXP \{ T_{i,i+1} \} + \sum_{i=2}^n \EXP \{ T_{i,i-1} \} \\
&= 2(n-1) + 2\sum_{i=2}^{n-1} \left( P_i + P_i P_{i-1} + \cdots + P_i \cdots P_{2} \right) + 2\sum_{i=2}^n \left( Q_i + Q_i Q_{i+1} + \cdots + Q_i \cdots Q_{n-1} \right) \\
&= 2(n-1) + 2\sum_{i=2}^{n-1} \sum_{j=2}^i \prod_{k=j}^i P_k + 2\sum_{i=2}^n \sum_{j=i}^{n-1} \prod_{k=i}^j Q_k  \\
&= 2(n-1) + 2\sum_{i=2}^{n-1} \sum_{j=2}^i \prod_{k=j}^i P_k + 2\sum_{j=2}^n \sum_{i=j}^{n-1} \prod_{k=j}^i Q_k  \\
&= 2(n-1) + 2\sum_{i=2}^{n-1} \sum_{j=2}^i \prod_{k=j}^i P_k + 2\sum_{i=2}^{n-1} \sum_{j=2}^i \prod_{k=j}^i Q_k  \\
&= 2(n-1) + 2\sum_{i=2}^{n-1} \sum_{j=2}^i \left( \prod_{k=j}^i P_k + \prod_{k=j}^i Q_k \right)  \\
&\ge 2(n-1) + 4\sum_{i=2}^{n-1} \sum_{j=2}^i \sqrt{ \prod_{k=j}^i P_k \times \prod_{k=j}^i Q_k } \\
&\qquad \hbox{\rm (by the arithmetic-geometric mean ineqality)} \\
&= 2(n-1) + 4\sum_{i=2}^{n-1} \sum_{j=2}^i  1 \\
&\qquad \hbox{\rm (since $P_iQ_i = 1$ for all $i$ in our range)} \\
&= 2(n-1) + 4\sum_{i=1}^{n-1} (i-1) \\
&= 2(n-1) + 2 n (n-1) - 4 (n-1) \\
&= 2(n-1)^2,\\
\end{aligned}
\]
which concludes the proof.
\end{proof}

Next we present a simple corollary of Lemma~\ref{lemma:backandforth} that is used in our lower bound.

\begin{corollary}\label{cor:markov}
Consider a random walk with reflecting barriers on $\{ -n, \ldots, n \}$,
$n > 0$. In this chain,
\[
\max \left( \EXP \{ T_{0,n} \} , \EXP \{ T_{0,-n} \} \right)
 \ge  {2 \over 3} \, n^2.
\]
\end{corollary}

\begin{proof}
We prove this by contradiction. Set $c = 2/3$. Assume that
\[
\max \left( \EXP \{ T_{0,n} \} , \EXP \{ T_{0,-n} \} \right) <  cn^2.
\]
By Theorem 1,
\[
 \EXP \{ T_{0,n} \}  +  \EXP \{ T_{n,0} \}  \ge 2n^2,
\]
and
\[
 \EXP \{ T_{0,-n} \}  +  \EXP \{ T_{-n,0} \}  \ge 2n^2.
\]
Observe for this that $0$ is not a reflecting barrier, but this makes 
$ \EXP \{ T_{0,n} \}$ only larger, so Theorem 1 does indeed apply.
By our assumption, we thus have
\[
\min \left( \EXP \{ T_{n,0} \} ,  \EXP \{ T_{-n,0} \} \right) >  (2-c) n^2.
\]
Let $T$ be the cover time, i.e., the time to visit all states starting from state $0$. It is easy to see that
\[
T_{0,n} + T_{0,-n} > T = \max \left( T_{0,n} , T_{0,-n} \right)
  =  T_{0,S} + T_{S,-S},
\]
where $S \in \{ n, -n \}$ is the first of the two end states
reached by the Markov chain. If we condition on the history up to $T_{0,S}$,
we see that
\[
\EXP \{ T_{S,-S} \} \ge \min \left( \EXP \{ T_{n,0} \} ,  \EXP \{ T_{-n,0} \} \right) >  (2-c) n^2.
\]
Thus,
\[
\max \left( \EXP \{ T_{0,n} \} , \EXP \{ T_{0,-n} \} \right)
 \ge {1 \over 2} \, \left( \EXP \{ T_{0,n} \} + \EXP \{ T_{0,-n} \} \right)
 \ge {1 \over 2} \, \EXP \left\{ \max \left( T_{0,n} , T_{0,-n} \right) \right\} 
= {1 \over 2} \, \EXP \{ T \}
>  (1- c/2 ) \,  n^2 ,
\]
which contradicts our assumption.
\end{proof}


\subsection{The Lower Bound}
\label{sec:bound}

Let $\mathcal{A}$ be a randomized memoryless routing algorithm.  Let $k$
be an even integer, let $t$ be the origin, and let $A=a_1,\ldots a_k$ be
a path of $k$ collinear vertices such that $a_{k}$ is closer to $t$ than
any of $a_1,\ldots,a_{k-1}$ and the three points $a_{1},a_{k},t$ make a left turn with $\angle
a_{1}a_{k}t$ greater than $150^\circ$ degrees but less than $180^\circ$ (see Figure~\ref{fig:ab}.a). Let
$B=b_1,\ldots,b_k$ be the reflection of $A$ through the line parallel
to $A$ that contains $t$ (see Figure~\ref{fig:ab}.b).  Let $A(\alpha)$,
respectively, $B(\alpha)$, denote the path $A$, respectively, $B$,
rotated by an angle of $\alpha$ about the origin, $t$.

\begin{figure}
  \begin{center}
    \begin{tabular}{cc}
      \includegraphics{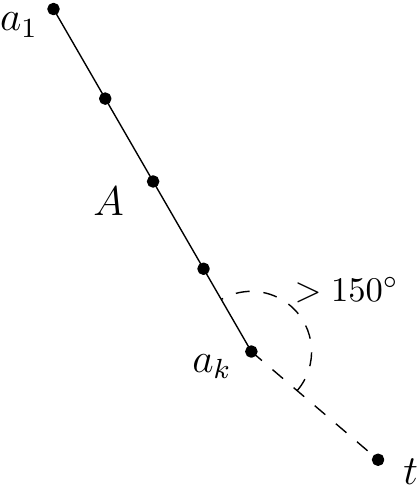} &
      \includegraphics{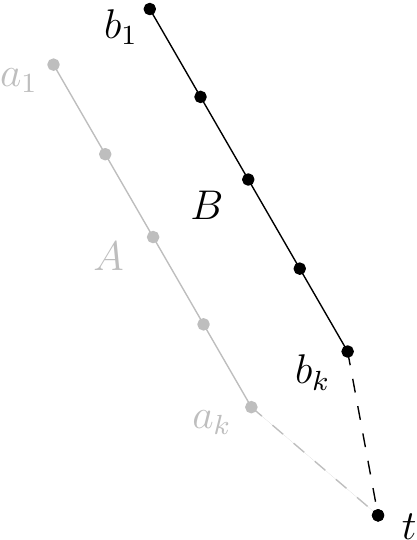} \\
      (a) & (b)
    \end{tabular}
  \end{center}
  \caption{The chains $A$ and $B$.}
  \label{fig:ab}
\end{figure}

Define the \emph{color} of a path $A(\alpha)=a_1',\ldots,a_k'$ as
follows: Imagine running $\mathcal{A}$ on the graph consisting of $A'$
and the isolated vertex $t$, starting at $a_{k/2}'$.  If $\mathcal{A}$
takes $\Omega (k^{2})$ expected time to reach $a_k'$ then color
$A(\alpha)$ \emph{blue}, otherwise color $A(\alpha)$ \emph{red}.
Note that Corollary~\ref{cor:markov} implies that, if $A(\alpha)$ is red,
then $\mathcal{A}$ takes $\Omega(k^2)$ expected time to reach $a_1$
starting at $a_{k/2}$.

Intuitively, a path is red (getting hotter --- closer to $t$) if
$\mathcal{A}$ could move quickly from $a_{k/2}$ to $a_k$.  A path is blue (getting cooler --- further from $t$) if $\mathcal{A}$ could move quickly to $a_1$.  Define the color (red or blue) of a path $B(\alpha)$ in the same way.

\begin{lemma}\label{lem:2blue}
If there exists $\alpha$ such that $A(\alpha)$ and $B(\alpha)$ are both blue, then there exists a convex subdivision $G=(V,E)$ with $|V|=2k+1$ with vertices $s,t\in V$ such that $\mathcal{A}$ takes $\Omega(k^2)$ steps when routing from $s$ to $t$.  
\end{lemma} 

\begin{proof}
Let $A'=A(\alpha)=a_1',\ldots,a_k'$ and $B'=B(\alpha)=b_1',\ldots,b_k'$.  The convex subdivision $G$ consists of $A'$ and $B'$ as well as the edges $a_1'b_1'$, $a_k't$ and $b_k't$ (see Figure~\ref{fig:2blue}).  Since $A(\alpha)$ and $B(\alpha)$ both blue, applying $\mathcal{A}$ to route from $a_{k/2}'$ to $t$ will require $\Omega(k^2)$ expected steps.
\end{proof}

\begin{figure}
  \centering
  \includegraphics{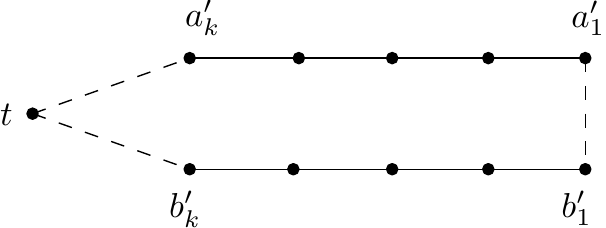}
  \caption{Two blue chains $A(\alpha)$ and $B(\alpha)$.}
  \label{fig:2blue}
\end{figure}

\begin{lemma}\label{lem:2red}
If there exists $\alpha$ such that $A(\alpha)$ and $B(180+\alpha)$ are both red, then there exists a convex subdivision $G=(V,E)$ with $|V|=2k+1$ with vertices $s,t\in V$ such that $\mathcal{A}$ takes $\Omega(k^2)$ steps when routing from $s$ to $t$.  
\end{lemma} 

\begin{proof}
Let $A'=A(\alpha)=a_1',\ldots,a_k'$ and $B'=B(180+\alpha)=b_1',\ldots,b_k'$.  The convex subdivision $G$ consists of $A'$ and $B'$ as well as the edges $a_k'b_k'$, $a_1't$ and $b_1't$ (see Figure~\ref{fig:2red}).  Since $A(\alpha)$ and $B(180+\alpha)$ are red, applying $\mathcal{A}$ to route from $a_{k/2}'$ to $t$ will require $\Omega(k^2)$ expected steps.
\end{proof}

\begin{figure}
  \centering
  \includegraphics{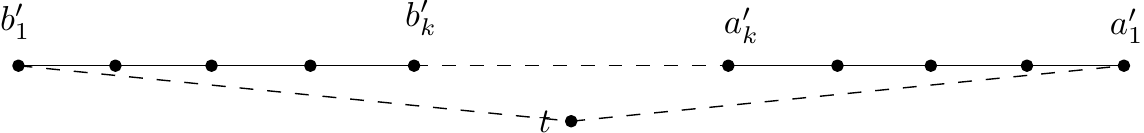}
  \caption{Two red chains $A(\alpha)$ and $B(180+\alpha)$.}
  \label{fig:2red}
\end{figure}

\begin{theorem}
  \label{thm:quadratic}
  For any integer $k> 0$ and any 
  memoryless routing algorithm $\mathcal{A}$,
  there exists a convex subdivision $G=(V,E)$ with $|V|=\Theta(k)$ having vertices $s,t\in V$ such that $\mathcal{A}$ takes $\Omega(k^2)$ steps when routing from $s$ to $t$.
\end{theorem}

\begin{proof}
If either of Lemma~\ref{lem:2blue} or Lemma~\ref{lem:2red} apply to $\mathcal{A}$ then the proof is complete.  Otherwise, observe that the exclusion of these two lemmata implies that, for any $\alpha$, at least one of $A(\alpha)$ and $B(\alpha)$ is blue.  To see this, note that if $A(\alpha)$ is red, then (the exclusion of) Lemma~\ref{lem:2red} implies that $B(\alpha+180)$ is blue, so (the exclusion of) Lemma~\ref{lem:2blue} implies that $A(\alpha+180)$ is red, so (the exclusion of) Lemma~\ref{lem:2red} implies that $B(\alpha)$ is blue.

Therefore, there exists 3 blue chains $X=x_1,\ldots,x_k$, $Y=y_1,\ldots,y_k$, and $Z=z_1,\ldots,z_k$ where $X\in\{A(0),B(0)\}$, $Y\in\{A(120),B(120)\}$ and $Z\in\{A(240),B(240)\}$.
We can then take $G$ to be the graph containing $X$, $Y$, and $Z$,
as well as the edges $x_1y_1$, $y_1z_1$, $z_1x_1$, $x_ky_k$, $y_kz_k$,
$z_kx_k$, $x_kt$, $y_kt$, $z_kt$ (see Figure~\ref{fig:3blue}).  Because
$X$, $Y$, and $Z$ are all blue, the expected number of steps required
to route from $x_{k/2}$ to $t$ using $\mathcal{A}$ is $\Omega(k^2)$.

\begin{figure}[ht]
  \centering
  \includegraphics{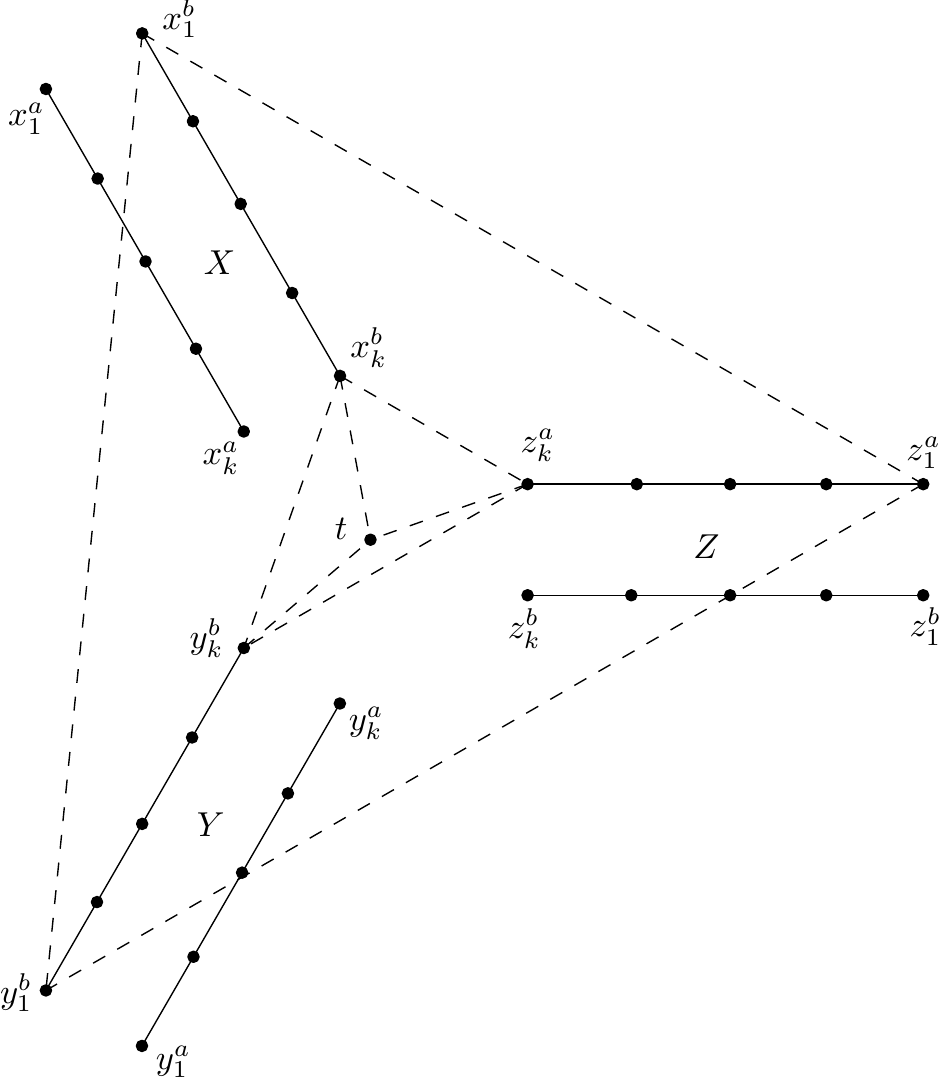}
  \caption{Three blue chains}
  \label{fig:3blue}
\end{figure}

All that remains is to verify that $G$ is indeed a convex subdivision.  This is
readily established using the fact that the angles $\angle x_1 x_k t$, $\angle y_1 y_k t$, and $\angle z_1 z_k t$, are all between 150 and 180 degrees.
%
%
%
%
%
%
\end{proof}

\section{Conclusions}

We have shown that the \textsc{Random-Compass} algorithm takes exponential
expected time to route on some convex subdivisions and that any randomized
memoryless routing algorithm takes at least quadratic time to route on some
convex subdivisions.  We conclude with two open problems:

\begin{op}
The current upper bound for the expected time required by \textsc{Random-Compass} on convex subdivisions is $O(2^n)$ and the lower bound is $\Omega(2^{n/3})$.  Close this gap.
\end{op}

\begin{op}
A random walk on $G$ routes any message in $O(n^2)$ expected time but requires $O(\log d)$ random bits when located at a vertex of degree $d$.  Is there a randomized memoryless routing algorithm for routing on convex subdivisions that uses $O(1)$ random bits per step and that routes any message in $O(n^2)$ expected time?
\end{op}

\bibliographystyle{plain}
\bibliography{convobl}
\end{document}